\documentclass[twocolumn, superscriptaddress,
showpacs,preprintnumbers,amsmath,amssymb]{revtex4}

\usepackage{amssymb,amsmath,amsthm}
\usepackage{multirow}
\newtheorem{Thm}{Theorem}
\newtheorem{Cor}{Corollary}

\theoremstyle{definition}

\newcommand{\bra}[1]{{\left\langle #1 \right|}}
\newcommand{\ket}[1]{{\left| #1 \right\rangle}}

\newcommand{\T}{\mbox{$\mathrm{tr}$}}

\newcommand{\B}{\mbox{$\mathbb B$}}

\begin{document}
\title{Weighted polygamy inequalities of multiparty entanglement in arbitrary dimensional quantum systems}

\author{Jeong San Kim}
\email{freddie1@khu.ac.kr} \affiliation{
 Department of Applied Mathematics and Institute of Natural Sciences, Kyung Hee University, Yongin-si, Gyeonggi-do 17104, Korea
}
\date{\today}

\begin{abstract}
We provide a generalization for the polygamy constraint of multiparty entanglement in arbitrary dimensional quantum systems.
By using the $\beta$th-power of entanglement of assistance for $0\leq \beta \leq1$ and the Hamming weight of the binary vector related with the distribution of subsystems, we establish a class of weighted polygamy inequalities of multiparty entanglement in arbitrary dimensional quantum systems. We further show that our class of weighted polygamy inequalities can even be improved to be tighter inequalities
with some conditions on the assisted entanglement of bipartite subsystems.
\end{abstract}

\maketitle

\section{Introduction}
\label{Intro}
One intrinsic feature of quantum entanglement is
the limited shareability of bipartite entanglement in multiparty quantum systems.
This distinct property of quantum entanglement without any classical counterpart is known as
the {\em monogamy of entanglement}(MoE)~\cite{T04, KGS}.

MoE is mathematically characterized in a quantitative way; for a given three-party quantum state $\rho_{ABC}$ with its reduced density matrices $\rho_{AB}=\T_C \rho_{ABC}$ and $\rho_{AC}=\T_B \rho_{ABC}$,
\begin{align}
E\left(\rho_{A|BC}\right)\geq E\left(\rho_{A|B}\right)+E\left(\rho_{A|C}\right)
\label{MoE}
\end{align}
where $E\left(\rho_{A|BC}\right)$ is the bipartite entanglement between subsystems $A$ and $BC$, and
$E\left(\rho_{A|B}\right)$ and $E\left(\rho_{A|C}\right)$ are the bipartite entanglement between $A$ and $B$
and between $A$ and $C$, respectively.
The {\em monogamy inequality} in~(\ref{MoE}) shows a mutually exclusive relation of the bipartite entanglement
between $A$ and each of $B$ and $C$(that is, $E\left(\rho_{A|B}\right)$ and $E\left(\rho_{A|C}\right)$, respectively),
so that their summation cannot exceeds the total entanglement between $A$ and $BC$(measured by $E\left(\rho_{A|BC}\right)$ ).

The first monogamy inequality was established in three-qubit systems using {\em tangle} as the bipartite entanglement measure~\cite{CKW}.
Later, it was generalized for multiqubit systems, and
some cases of higher-dimensional quantum systems in terms of various bipartite entanglement measures~\cite{OV, KS, KDS, KSRenyi, KT, KSU}.

Whereas MoE reveals the limited shareability of entanglement in multiparty quantum systems,
the {\em assisted entanglement}, which is a dual amount to bipartite entanglement measures, is also known to have
a dually monogamous property in multiparty quantum systems, namely, {\em polygamy of entanglement}(PoE).
PoE is also mathematically characterized as {\em polygamy inequality};
\begin{align}
E_a\left(\rho_{A|BC}\right)\leq E_a\left(\rho_{A|B}\right)
+E_a\left(\rho_{A|C}\right),
\label{PoE}
\end{align}
for a three-party quantum state $\rho_{ABC}$ where $E_a\left(\rho_{A|BC}\right)$ is the assisted entanglement~\cite{GMS}.

The polygamy inequality in~(\ref{PoE}) was first proposed in three-qubit systems using {\em tangle of assistance}~\cite{GMS},
and generalized into multiqubit systems in terms of various assisted entanglements~\cite{GBS, KT, KSU}.
For quantum systems beyond qubits, a general polygamy inequality of multiparty entanglement in arbitrary
dimensional quantum systems was established using entanglement of assistance~\cite{BGK, KimGP, KimGP16}.

One main difficulty in studying entanglement in multiparty quantum systems is that
there are several inequivalent classes of genuine multiparty quantum entanglement that are not convertible to each other
by means of stochastic local operations and classical communications(SLOCC)~\cite{DVC};
for example, there are two inequivalent classes of genuine three-party pure entangled states in three-qubit systems~\cite{DVC}. One is
the {\em Greenberger-Horne-Zeilinger}(GHZ) class~\cite{GHZ},
and the other one is the W-class~\cite{DVC}.
The existence of these inequivalent classes make us infeasible to directly compare the amount of entanglement from different classes,
which also implies the hardness of having a universal way to quantify multiparty quantum entanglement, even abstractly.

Although this characterization is due to the interconvertibility under SLOCC,
these inequivalent classes of genuine three-qubit entangled states also reveal different characters in terms of entanglement
monogamy and polygamy. The tangle-based monogamy and polygamy inequalities of three-qubit entanglement in Inequalities (\ref{MoE})
and (\ref{PoE}) are saturated (thus they hold as equalities) by the W-class states,
whereas the differences between terms can assume their largest values for the GHZ-class states.

The saturation of the monogamy and polygamy inequalities for W-class states implies that
this type of genuine three-qubit entanglement can be complete characterized by means of the bipartite ones within it,
which is not the case for the GHZ-class states, the other type of genuine three-qubit entanglement.
Thus entanglement monogamy and polygamy are not just distinct phenomena in multipartite quantum systems, but they also provide us an efficient way to qualify multipartite entanglements from different classes.

For the case of multi-qubit W-class states more that three qubits, the tangle-based monogamy and polygamy inequalities are also saturated
by this class, and thus an analogous interpretation can be applied. However, tangle is known to fail in generalizing the monogamy inequality into higher-dimensional systems more than qubits~\cite{KDS}. This imposes the importance of having proper bipartite entanglement quantifications
showing tight monogamy and polygamy inequalities for an efficient characterization of multiparty entanglements
from different classes even in high-dimensional quantum systems.

Recently, monogamy and polygamy inequalities of multiqubit entanglement were generalized
in terms of non-negative power of entanglement measures and assisted entanglements;
it was shown that the $\alpha$th-power of the entanglement of formation and concurrence can be used to establish multiqubit monogamy inequalities for $\alpha \geq \sqrt{2}$ and $\alpha \geq 2$, respectively~\cite{Fei}.
Later, tight classes of monogamy and polygamy inequalities of multiqubit entanglement using non-negative power of various entanglement measures were also proposed~\cite{Fei2, KimPS18, KimPU18}. However, the validity of this tight generalization of entanglement constraints beyond qubit systems is still unclear.

Here, we provide a tight polygamy constraint of multiparty entanglement in arbitrary dimensional quantum systems.
By using the $\beta$th-power of entanglement of assistance for $0\leq \beta \leq1$ and the Hamming weight of the binary vector related with the distribution of subsystems, we establish a class of weighted polygamy inequalities of multiparty entanglement in arbitrary dimensional quantum systems. We further show that our class of weighted polygamy inequalities can even be improved to be tighter inequalities
with some conditions on the assisted entanglement of bipartite subsystems.

The paper is organized as follows. In Sec.~\ref{Sec: poly}, we review the polygamy constraints of multiparty quantum entanglement
based on tangle and entanglement of assistance.
In Sec.~\ref{Sec: WPoly}, we first provide some notations and
definitions about binary vectors as well as its Hamming weight, and establish a class of weighted polygamy inequalities
of multiparty entanglement using the $\beta$th-power of entanglement of assistance for $0\leq \beta \leq1$.
We also show that our class of weighted polygamy inequalities can be improved to be tighter inequalities
with some conditions on the assisted entanglement of bipartite subsystems.
Finally, we summarize our results in Sec.~\ref{Sec: Conclusion}.

\section{Polygamy of multiparty Quantum Entanglement}
\label{Sec: poly}

The first polygamy inequality was established in three-qubit systems~\cite{GMS};
for a three-qubit pure state $\ket{\psi}_{ABC}$,
\begin{equation}
\tau\left(\ket{\psi}_{A|BC)}\right)\le\tau_a\left(\rho_{A|B}\right)
+\tau_a\left(\rho_{A|C}\right), \label{3dual}
\end{equation}
where
\begin{align}
\tau\left(\ket{\psi}_{A|BC}\right)=4\det\rho_A
\label{tangle}
\end{align}
is the tangle of the pure state $\ket{\psi}_{ABC}$ between $A$ and $BC$, and
\begin{align}
\tau_a\left(\rho_{A|B}\right)=\max \sum_i p_i\tau\left({\ket{\psi_i}_{A|B}}\right)
\label{tanglemix}
\end{align}
is the tangle of assistance of $\rho_{AB}=\T_C\ket{\psi}_{ABC}\bra{\psi}$ with the maximum taken over all
possible pure-state decompositions of $\rho_{AB}=\sum_{i}p_i\ket{\psi_i}_{AB}\bra{\psi_i}$.
Later, Inequality~(\ref{3dual}) was generalized into multiqubit
systems~\cite{GBS}
\begin{align}
\tau_a\left(\rho_{A_1|A_2\cdots A_n}\right)
\leq&\sum_{i=2}^{n}\tau_a \left(\rho_{A_1|A_i}\right),
\label{ntpoly}
\end{align}
for an arbitrary multiqubit mixed state $\rho_{A_1\cdots A_n}$ and its two-qubit reduced density matrices $\rho_{A_1A_i}$
with $i=2,\ldots, n$.

For polygamy inequality beyond qubits, it was shown that von Neumann entropy can be used
to establish a polygamy inequality of three-party quantum systems~\cite{BGK};
for any three-party pure state $\ket{\psi}_{ABC}$ of arbitrary dimensions, we have
\begin{align}
E\left(\ket{\psi}_{A|BC}\right)\leq& E_a\left(\rho_{A|B}\right)+E_a\left(\rho_{A|C}\right),
\label{EoApoly3}
\end{align}
where
\begin{equation}
E\left(\ket{\psi}_{A|BC}\right)=S\left(\rho_A\right)
\label{eent}
\end{equation}
is the entropy of entanglement between $A$ and $BC$ in terms of the von Neumann entropy
\begin{equation}
S(\rho)=-\T \rho\log\rho,
\label{von}
\end{equation}
and $E_a(\rho_{A|B})$ is the entanglement of assistance(EoA)
of $\rho_{AB}$ defined as~\cite{cohen}
\begin{equation}
E_a(\rho_{A|B})=\max \sum_{i}p_i E\left(\ket{\psi_i}_{A|B}\right)
\label{eoa}
\end{equation}
with the maximization over all possible pure state
decompositions of  $\rho_{AB}=\sum_{i}p_i\ket{\psi_i}_{AB}\bra{\psi_i}$.
Later, a general polygamy inequality of multiparty quantum entanglement was established as
\begin{align}
E_a\left(\rho_{A_1|A_2\cdots A_n}\right)
\leq& \sum_{i=2}^{n} E_a \left(\rho_{A_1|A_i}\right),
\label{EoApoly}
\end{align}
for any multiparty quantum state $\rho_{A_1A_2\cdots A_n}$ of arbitrary dimension~\cite{KimGP}.

\section{Weighted Polygamy Constraints of multiparty Quantum Entanglement}
\label{Sec: WPoly}
Based on the binary expression of any nonnegative integer $j$,
\begin{align}
j=\sum_{i=0}^{n-1} j_i 2^i
\label{bire2}
\end{align}
such that $\log_{2}j \leq n$ and $j_i \in \{0, 1\}$ for $i=0, \ldots, n-1$, we define a unique binary vector
$\overrightarrow{j}$ associated with $j$ as
\begin{align}
\overrightarrow{j}=\left(j_0,~j_1,\ldots,j_{n-1}\right).
\label{bivec}
\end{align}
For the binary vector $\overrightarrow{j}$ in  Eq.~(\ref{bivec}), its {\em Hamming weight}~\cite{nc},
$\omega_{H}\left(\overrightarrow{j}\right)$, is defined as the number of $1's$ in its coordinates, that is, the number of $1's$ in $\{ j_0,~j_1,\ldots,j_{n-1} \}$.

The following theorem states that a class of weighted polygamy inequalities of multiparty entanglement in arbitrary dimension can be established using the $\beta$th-power of EoA and the Hamming weight of the binary vector related with the distribution of subsystems.

\begin{Thm}
For $0 \leq \beta \leq 1$ and any $N+1$-party quantum state $\rho_{A\B}$ where $\B$ consists of $N$-party subsystems,
there exists a proper ordering of
the $N$-party subsystems $\B=\{B_0, \cdots, B_{N-1}\}$ such that
\begin{equation}
\left(E_a\left(\rho_{A|B_0 B_1 \cdots B_{N-1}}\right)\right)^{\beta}  \leq
\sum_{j=0}^{N-1}{\beta}^{\omega_{H}\left(\overrightarrow{j}\right)}\left(E_a\left(\rho_{A|B_j}\right)\right)^{\beta}.
\label{peoapoly0}
\end{equation}
\label{tpoly0}
\end{Thm}

\begin{proof}
Let us consider the ordering of the $N$-party subsystems $\B=\{B_0, \cdots, B_{N-1}\}$ where the EoA's between $A$ and each $B_j$ are in decreasing order, that is,
\begin{align}
E_a\left(\rho_{A |B_j}\right)\geq E_a\left(\rho_{A |B_{j+1}}\right)\geq0
\label{pordern2}
\end{align}
for each $j=0, \ldots , N-2$.

From the monotonicity of the function $f(x)=x^{\beta}$ for $0 \leq \beta \leq 1$ and
Inequality~(\ref{EoApoly}), we have
\begin{align}
\left(E_a\left(\rho_{A|B_0B_1\cdots B_{N-1}}\right)\right)^{\beta}\leq&
\left(\sum_{j=0}^{N-1}E_a\left(\rho_{A|B_j}\right)\right)^{\beta},
\label{qbitalppoly1}
\end{align}
therefore, it is enough to show that
\begin{align}
\left(\sum_{j=0}^{N-1}E_a\left(\rho_{A|B_j}\right)\right)^{\beta}\leq&
\sum_{j=0}^{N-1}{\beta}^{\omega_{H}\left(\overrightarrow{j}\right)}\left(E_a\left(\rho_{A|B_j}\right)\right)^{\beta}.
\label{qbitalppoly2}
\end{align}

We first prove Inequality~(\ref{qbitalppoly2}) for the case that $N=2^n$, a power of 2, by using mathematical induction on $n$, and extend the result for any positive integer $N$.
For $n=1$ and a three-party state $\rho_{AB_0B_1}$ with bipartite reduced density matrices $\rho_{AB_0}$ and $\rho_{AB_1}$,
we have
\begin{align}
&\left(E_a\left(\rho_{A|B_0}\right)+E_a\left(\rho_{A|B_1}\right)\right)^{\beta}\nonumber\\
&~~~~~~~~~~=\left(E_a\left(\rho_{A|B_0}\right)\right)^{\beta}\left(1+\frac{E_a\left(\rho_{A|B_1}\right)}
{E_a\left(\rho_{A|B_0}\right)} \right)^{\beta}.
\label{p3scpoly1}
\end{align}

Because the ordering in Inequality~(\ref{pordern2}) assures~\cite{zerook}
\begin{equation}
0\leq\frac{E_a\left(\rho_{A|B_1}\right)}{E_a\left(\rho_{A|B_0}\right)}\leq 1,
\label{conoder0}
\end{equation}
Eq.~(\ref{p3scpoly1}) leads us to
\begin{align}
&\left(E_a\left(\rho_{A|B_0}\right)+E_a\left(\rho_{A|B_1}\right)\right)^{\beta}\nonumber\\
&~~~~~~~~~~~\leq\left(E_a\left(\rho_{A|B_0}\right)\right)^{\beta}+
\beta\left(E_a\left(\rho_{A|B_1}\right)\right)^{\beta},
\label{p3scpoly3}
\end{align}
where the inequality is due to
\begin{align}
\left(1+x\right)^{\beta}\leq 1+\beta x^{\beta},
\label{betle1}
\end{align}
for any $x \in \left[0,1\right]$ and $0\leq\beta \leq1$.
Inequality~(\ref{p3scpoly3}) recovers Inequality~(\ref{qbitalppoly2}) for $N=2$, that is, $n=1$.

Now we assume the validity of Inequality~(\ref{qbitalppoly2}) for $N=2^{n-1}$ with $n\geq 2$, and consider the case that $N=2^n$.
For an $(N+1)$-party quantum state $\rho_{AB_0B_1 \cdots B_{N-1}}$ and its bipartite reduced density matrices $\rho_{AB_j}$ with $j=0, \ldots, N-1$, the ordering of subsystems in Inequality~(\ref{pordern2}) assures that
\begin{equation}
0\leq\frac{\sum_{j=2^{n-1}}^{2^n-1}E_a\left(\rho_{A|B_j}\right)}
{\sum_{j=0}^{2^{n-1}-1}E_a\left(\rho_{A|B_j}\right)}\leq 1.
\label{conoder1}
\end{equation}
Thus we have
\begin{widetext}
\begin{align}
\left(\sum_{j=0}^{2^n-1}E_a\left(\rho_{A|B_j}\right)\right)^{\beta}
=&\left(\sum_{j=0}^{2^{n-1}-1}E_a\left(\rho_{A|B_j}\right)\right)^{\beta}
\left(1+\frac{\sum_{j=2^{n-1}}^{2^n-1}E_a\left(\rho_{A|B_j}\right)}
{\sum_{j=0}^{2^{n-1}-1}E_a\left(\rho_{A|B_j}\right)}\right)^{\beta}\nonumber\\
\leq& \left(\sum_{j=0}^{2^{n-1}-1}E_a\left(\rho_{A|B_j}\right)\right)^{\beta}
\left[1+\beta\left(\frac{\sum_{j=2^{n-1}}^{2^n-1}E_a\left(\rho_{A|B_j}\right)}
{\sum_{j=0}^{2^{n-1}-1}E_a\left(\rho_{A|B_j}\right)}\right)^{\beta}\right]\nonumber\\
=&\left(\sum_{j=0}^{2^{n-1}-1}E_a\left(\rho_{A|B_j}\right)\right)^{\beta}
+\beta\left(\sum_{j=2^{n-1}}^{2^{n}-1}E_a\left(\rho_{A|B_j}\right)\right)^{\beta},
\label{pnscpoly}
\end{align}
\end{widetext}
where the inequality is due to Inequality~(\ref{betle1}).

From the induction hypothesis, we have
\begin{align}
\left(\sum_{j=0}^{2^{n-1}-1}E_a\left(\rho_{A|B_j}\right)\right)^{\beta}\leq&
\sum_{j=0}^{2^{n-1}-1}{\beta}^{\omega_{H}\left(\overrightarrow{j}\right)}\left(E_a\left(\rho_{A|B_j}\right)\right)^{\beta}.
\label{pnscpoly2}
\end{align}
Moreover, the second summation in the last line of~(\ref{pnscpoly}) is a summation of $2^{n-1}$ terms,
therefore the induction hypothesis also guarantees
\begin{align}
\left(\sum_{j=2^{n-1}}^{2^{n}-1}E_a\left(\rho_{A|B_j}\right)\right)^{\beta}\leq&
\sum_{j=2^{n-1}}^{2^{n}-1}{\beta}^{\omega_{H}\left(\overrightarrow{j}\right)-1}\left(E_a\left(\rho_{A|B_j}\right)\right)^{\beta}.
\label{pnscpoly3}
\end{align}
(Possibly, we may index and reindex subsystems to get Inequality~(\ref{pnscpoly3}), if necessary.)

From Inequalities~(\ref{pnscpoly}),~ (\ref{pnscpoly2}) and (\ref{pnscpoly3}), we have
\begin{align}
\left(\sum_{j=0}^{2^n-1}E_a\left(\rho_{A|B_j}\right)\right)^{\beta}\leq&
\sum_{j=0}^{2^n-1}{\beta}^{\omega_{H}\left(\overrightarrow{j}\right)}\left(E_a\left(\rho_{A|B_j}\right)\right)^{\beta},
\label{qditp2poly}
\end{align}
which recovers Inequality~(\ref{qbitalppoly2}) for the case that $N=2^n$.

Now let us consider an arbitrary positive integer $N$ and a $(N+1)$-party quantum state $\rho_{AB_0B_1\cdots B_{N-1}}$. We first note that we can always consider a power of $2$ that is an upper bound of $N$, that is $0\leq N \leq 2^{n}$ for some $n$. We also consider a $(2^{n}+1)$-party
quantum state
\begin{align}
\gamma_{AB_0 B_1 \cdots B_{2^n-1}}=\rho_{AB_0B_1\cdots B_{N-1}}\otimes\sigma_{B_N \cdots B_{2^n-1}},
\label{gamma}
\end{align}
which is a product of $\rho_{AB_0B_1\cdots B_{N-1}}$ and an arbitrary $(2^n-N)$-party quantum state $\sigma_{B_N \cdots B_{2^n-1}}$.

Because $\gamma_{AB_0 B_1 \cdots B_{2^n-1}}$ is a $(2^{n}+1)$-party quantum state, Inequality~(\ref{qditp2poly}) leads us to
\begin{equation}
\left(E_a\left(\gamma_{A|B_0 B_1 \cdots B_{2^n-1}}\right)\right)^{\beta}  \leq
\sum_{j=0}^{2^n-1}{\beta}^{\omega_{H}\left(\overrightarrow{j}\right)}
\left(E_a\left(\gamma_{A|B_j}\right)\right)^{\beta},
\label{gapoly}
\end{equation}
where $\gamma_{AB_j}$ is the bipartite reduced density matric of $\gamma_{AB_0 B_1 \cdots B_{2^n-1}}$ for each $j= 0, \ldots, 2^n-1$.
Moreover, $\gamma_{AB_0 B_1 \cdots B_{2^n-1}}$ is a product state of $\rho_{AB_0B_1\cdots B_{N-1}}$ and $\sigma_{B_N \cdots B_{2^n-1}}$,
which implies
\begin{align}
E_a\left(\gamma_{A|B_0 B_1 \cdots B_{2^n-1}}\right)=E_a\left(\rho_{A|B_0 B_1 \cdots B_{N-1}}\right),
\label{psame1}
\end{align}
and
\begin{align}
E_a\left(\gamma_{A|B_j}\right)=0,
\label{psame3}
\end{align}
for $j=N, \ldots , 2^n-1$.
Because
\begin{align}
\gamma_{AB_j}=\rho_{AB_j},
\label{psame2}
\end{align}
for each $j=0, \ldots , N-1$,
we have
\begin{align}
\left(E_a\left(\rho_{A|B_0 B_1 \cdots B_{N-1}}\right)\right)^{\beta}=&
\left(E_a\left(\gamma_{A|B_0 B_1 \cdots B_{2^n-1}}\right)\right)^{\beta}\nonumber\\
\leq&
\sum_{j=0}^{2^n-1}{\beta}^{\omega_{H}\left(\overrightarrow{j}\right)}
\left(E_a\left(\gamma_{A|B_j}\right)\right)^{\beta}\nonumber\\
=&\sum_{j=0}^{N-1}{\beta}^{\omega_{H}\left(\overrightarrow{j}\right)}
\left(E_a\left(\rho_{A|B_j}\right)\right)^{\beta},
\label{gapoly2}
\end{align}
and this completes the proof.
\end{proof}

To illustrate the tightness of Inequality~(\ref{peoapoly0}) compared with Inequality~(\ref{EoApoly}) in previous section,
let us consider the three-qubit W state
\begin{align}
\ket{W}_{ABC} =\frac{1}{\sqrt 3}\left(\ket{100}+\ket{010}+\ket{001}\right).
\label{3W}
\end{align}
Because it is a pure state, we have
\begin{align}
E_a\left(\rho_{A|BC}\right)=&S(\rho_A)
=\log3-\frac{2}{3},
\label{EoA1}
\end{align}
and the EoA of the two-qubit reduced density matrices are~\cite{sahoo}
\begin{align}
E_a\left(\rho_{A|B}\right)=E_a\left(\rho_{A|C}\right)=\frac{2}{3}.
\label{EoA123}
\end{align}
Thus, the marginal EoA from Inequality~(\ref{EoApoly}) is
\begin{align}
E_a\left(\rho_{A|B}\right)&+E_a\left(\rho_{A|C}\right)\nonumber\\
&-E_a\left(\rho_{A|BC}\right)=2-\log3 \approx 0.415.
\label{mareoa}
\end{align}

For the cases that $\beta=\frac{1}{2}$ or $\frac{1}{3}$, the marginal EoA's from Inequality~(\ref{peoapoly0}) for three-qubit W state are
\begin{align}
\sqrt{E_a\left(\rho_{A|B}\right)}&+\frac{1}{2}\sqrt{E_a\left(\rho_{A|C}\right)}\nonumber\\
&-\sqrt{E_a\left(\rho_{A|BC}\right)}\approx 0.272,\nonumber\\
E_a\left(\rho_{A|B}\right)^{1/3}&+\frac{1}{3}E_a\left(\rho_{A|C}\right)^{1/3}\nonumber\\
&-E_a\left(\rho_{A|BC}\right)^{1/3}\approx 0.196.
\label{mareoa1/2}
\end{align}
Thus Inequality~(\ref{peoapoly0}) is generally tighter than Inequality~(\ref{EoApoly}), which also delivers better bounds to characterize
the W-class type three-party entanglement by means of bipartite ones.

For any $0\leq \beta \leq 1$ and the Hamming weight $\omega_{H}\left(\overrightarrow{j}\right)$ of the binary vector $\overrightarrow{j}=\left(j_0, \ldots ,j_{n-1}\right)$, we have $0\leq {\beta}^{\omega_{H}\left(\overrightarrow{j}\right)}\leq1$, therefore
\begin{align}
\left(E_a\left(\rho_{A|B_0 B_1 \cdots B_{N-1}}\right)\right)^{\beta} \leq&
\sum_{j=0}^{N-1}{\beta}^{\omega_{H}\left(\overrightarrow{j}\right)}\left(E_a\left(\rho_{A|B_j}\right)\right)^{\beta}\nonumber\\
\leq&\sum_{j=0}^{N-1}\left(E_a\left(\rho_{A|B_j}\right)\right)^{\beta},
\label{ineqtight1}
\end{align}
for any multiparty state $\rho_{AB_0 B_1 \cdots B_{N-1}}$. Thus we have the following corollary;
\begin{Cor}
For $0 \leq \beta \leq 1$ and any multiparty quantum state $\rho_{AB_0\cdots B_{N-1}}$, we have
\begin{equation}
\left(E_a\left(\rho_{A|B_0 B_1 \cdots B_{N-1}}\right)\right)^{\beta}  \leq
\sum_{j=0}^{N-1}\left(E_a\left(\rho_{A|B_j}\right)\right)^{\beta}.
\label{peoapoly1}
\end{equation}
\label{Cor: poly}
\end{Cor}

We further note that the class of weighted polygamy inequalities in Theorem~\ref{tpoly0} can even be tightened
with some condition on bipartite entanglement of assistance.
\begin{Thm}
For $0 \leq \beta \leq 1$ and any multiparty quantum state $\rho_{AB_0 \cdots B_{N-1}}$, we have
\begin{equation}
\left(E_a\left(\rho_{A|B_0 \cdots B_{N-1}}\right)\right)^{\beta}  \leq
\sum_{j=0}^{N-1}{\beta}^{j}\left(E_a\left(\rho_{A|B_j}\right)\right)^{\beta},
\label{peoapoly2}
\end{equation}
conditioned that
\begin{align}
E_a\left(\rho_{A|B_i}\right)\geq \sum_{j=i+1}^{N-1}E_a\left(\rho_{A|B_{j}}\right),
\label{cond3}
\end{align}
for $i=0, \ldots , N-2$.
\label{tpoly2}
\end{Thm}

\begin{proof}
Inequality~(\ref{qbitalppoly1}) assures that it is enough to show
\begin{align}
\left(\sum_{j=0}^{N-1}E_a\left(\rho_{A|B_j}\right)\right)^{\beta}\leq&
\sum_{j=0}^{N-1}{\beta}^{j}\left(E_a\left(\rho_{A|B_j}\right)\right)^{\beta}.
\label{qditalppoly3}
\end{align}
We use the mathematical induction on $N$, and we also note that
Inequality~(\ref{p3scpoly3}) guarantees the validity of Inequality~(\ref{qditalppoly3}) for $N=2$.

Let us assume Inequality~(\ref{qditalppoly3}) is true for any positive integer less than $N$, and consider
a multiparty quantum state $\rho_{AB_0 \cdots B_{N-1}}$.
The condition in Inequality~(\ref{cond3}) assures
\begin{equation}
0\leq\frac{\sum_{j=1}^{N-1}E_a\left(\rho_{A|B_j}\right)}
{E_a\left(\rho_{A|B_0}\right)}\leq 1,
\label{conoder2}
\end{equation}
thus, we have
\begin{widetext}
\begin{align}
\left(\sum_{j=0}^{N-1}E_a\left(\rho_{A|B_j}\right)\right)^{\beta}
=&\left(E_a\left(\rho_{A|B_0}\right)\right)^{\beta}
\left(1+\frac{\sum_{j=1}^{N-1}E_a\left(\rho_{A|B_j}\right)}
{E_a\left(\rho_{A|B_0}\right)} \right)^{\beta}\nonumber\\
\leq&\left(E_a\left(\rho_{A|B_0}\right)\right)^{\beta}
\left[1+\beta\left(\frac{\sum_{j=1}^{N-1}E_a\left(\rho_{A|B_j}\right)}
{E_a\left(\rho_{A|B_0}\right)}\right)^{\beta}\right]\nonumber\\
=&\left(E_a\left(\rho_{A|B_0}\right)\right)^{\beta}+\beta\left(\sum_{j=1}^{N-1}E_a\left(\rho_{A|B_j}\right)
\right)^{\beta},
\label{pnscpoly6}
\end{align}
\end{widetext}
where the inequality is due to Inequality~(\ref{betle1}).

The summation in the last line of~(\ref{pnscpoly6}) is a summation of $N-1$ terms, therefore
the induction hypothesis leads us to
\begin{align}
\left(\sum_{j=1}^{N-1}E_a\left(\rho_{A|B_j}\right)
\right)^{\beta}\leq
\sum_{j=1}^{N-1}{\beta}^{j-1}\left(E_a\left(\rho_{A|B_j}\right)\right)^{\beta}.
\label{pnscpoly7}
\end{align}
Now, Inequalities~(\ref{pnscpoly6}) and (\ref{pnscpoly7}) recover Inequality~(\ref{qditalppoly3}),
and this completes the proof.
\end{proof}

For any nonnegative integer $j$ and its corresponding binary vector $\overrightarrow{j}$, the Hamming weight $\omega_{H}\left(\overrightarrow{j}\right)$ is bounded above by $\log_2 j$. Thus we have
\begin{align}
\omega_{H}\left(\overrightarrow{j}\right)\leq \log_2 j \leq j,
\label{numcom}
\end{align}
therefore
\begin{align}
\left(E_a\left(\ket{\psi}_{A|B_0 \cdots B_{N-1}}\right)\right)^{\beta}  \leq&
\sum_{j=0}^{N-1}{\beta}^{j}\left(E_a\left(\rho_{A|B_j}\right)\right)^{\beta}\nonumber\\
\leq& \sum_{j=0}^{N-1}{\beta}^{\omega_{H}\left(\overrightarrow{j}\right)}\left(E_a\left(\rho_{A|B_j}\right)\right)^{\beta},
\label{ineqtight3}
\end{align}
for $0\leq \beta \leq 1$. Thus, Inequality~(\ref{peoapoly2}) of Theorem~\ref{tpoly2} is tighter than Inequality~(\ref{peoapoly0}) of
Theorem~\ref{tpoly0} for $0\leq \beta \leq 1$ and any multiparty quantum state $\rho_{AB_0 B_1 \cdots B_{N-1}}$ satisfying the condition in Inequality~(\ref{cond3}).

\section{Conclusions}\label{Sec: Conclusion}
We have provided a generalization for the polygamy constraint of multiparty entanglement in arbitrary dimensional quantum systems.
By using the $\beta$th-power of entanglement of assistance for $0\leq \beta \leq1$ and the Hamming weight of the binary vector related with the distribution of subsystems, we have establish a class of weighted polygamy inequalities of multiparty entanglement in arbitrary dimensional quantum systems. We have further shown that our class of weighted polygamy inequalities can be improved to be tighter inequalities with some conditions on the assisted entanglement of bipartite subsystems.

The study of higher-dimensional quantum systems is important and even necessary in various
quantum information and communication processing tasks. For instance,
qudit systems for $d>2$ are sometimes preferred in quantum cryptography such as in quantum key distribution where
the use of qudits increases coding density and provides stronger security compared to qubits~\cite{gjvw}.

However, the entanglement properties in higher-dimensional systems are hardly known so far, and
the generalization of the multiparty entanglement analysis, especially the monogamy and polygamy constraints from qubit to qudit case is far
more than trivial. Thus even fundamental steps of the challenges to the richness of entanglement
studies for system of multiparty higher-dimensions systems would be necessary and fruitful to understand the whole picture of quantum entanglement.

Our results presented here deal with a generalized polygamy constraints of multyparty entanglement in arbitrary higher dimensional
quantum systems. Moreover, our class of polygamy inequalities provide tighter constraints which can also provide finer
characterizations of the entanglement distributions among the multiparty systems. Noting the importance of the study on multiparty quantum entanglement especially in higher dimensional quantum systems, our result can provide a rich reference for future work on the study of multiparty quantum entanglement.

\section*{Acknowledgments}
This work was supported by Basic Science Research Program through the National Research Foundation of Korea(NRF)
funded by the Ministry of Education(NRF-2017R1D1A1B03034727) and a grant from Kyung Hee University in 2017(KHU-20170716).



\begin{thebibliography}{1}
\bibitem{zerook}
If $E_a\left(\rho_{A|B_0}\right)=0$, the condition in Inequality~(\ref{pordern2}) implies $E_a\left(\rho_{A|B_j}\right)=0$ for each $j=1,\ldots, N-1$, and this makes Inequality~(\ref{peoapoly0}) trivial as an equality. For this reason, we may consider the case $E_a\left(\rho_{A|B_0}\right)>0$ so that the fraction in Inequality~(\ref{conoder0}) is well defined.
\end{thebibliography}
\end{document}